\documentclass{article}
\usepackage{graphicx} 
\usepackage{mathtools,amsmath,amsthm}
\usepackage{amssymb}
\usepackage{authblk}
\usepackage{accents,url,hyperref} 
\usepackage{accents,bm,bbm,braket,mleftright,mathtools,overpic,wasysym,xcolor}
\graphicspath{{./figures/}}
\usepackage{endnotes}
\let\footnote=\endnote

\usepackage{amsthm}
\usepackage{thmtools, thm-restate} 
\declaretheorem[name=Theorem]{theorem}

\declaretheorem[name=Resource, sibling=theorem]{resource}
\declaretheorem[name=Lemma, sibling=theorem]{lemma}

\declaretheorem[name=Definition, sibling=theorem, style=definition]{definition}
\declaretheorem[name=Remark, sibling=theorem, style=definition]{remark}

\declaretheorem[name=Example, sibling=theorem, style=definition]{example}

\newcommand{\CC}{\mathbb{C}}
\newcommand{\NN}{\mathbb{N}}

\newcommand{\eps}{\varepsilon}
\newcommand{\GL}{\mathrm{GL}}

\newcommand{\bH}{\mathbf{H}}
\newcommand{\bdisj}{\mathbf{Disjoint}}
\newcommand{\blattice}{\mathbf{Lattice}}
\newcommand{\bfan}{\mathbf{Fan}}
\newcommand{\bkagome}{\mathbf{Kagome}}
\newcommand{\bStr}{\mathbf{Strassen}}
\newcommand{\disj}{\mathrm{Disjoint}}
\newcommand{\fan}{\mathrm{Fan}}
\newcommand{\lattice}{\mathrm{Lattice}}
\newcommand{\kagome}{\mathrm{Kagome}}
\newcommand{\Str}{\mathrm{Strassen}}
\newcommand{\Young}{R}

\newcommand{\GHZ}{\operatorname{GHZ}}
\newcommand{\EPR}{\operatorname{EPR}}
\newcommand{\W}{\operatorname{W}}

\newcommand{\rk}{\operatorname{rk}}
\newcommand{\id}{\mathbbm{1}}
\newcommand{\ot}{\otimes}

\newcommand{\geqdeg}{\trianglerighteq}
\newcommand{\geqas}{\gtrsim}

\newcommand{\asymprank}{\underaccent{\tilde}{R}}
\newcommand{\asympsub}{\underaccent{\tilde}{Q}}

\title{The Tensor as an Informational Resource}
\author{Matthias Christandl\thanks{christandl@math.ku.dk}}
\affil{University of Copenhagen, Copenhagen, Denmark}
\date{16 September 2024}

\begin{document}

\maketitle

\begin{abstract}
A tensor is a multidimensional array of numbers that can be used to store data, encode a computational relation and represent quantum entanglement. In this sense a tensor can be viewed as valuable resource whose transformation can lead to an understanding of structure in data, computational complexity and quantum information. 

\sloppy
In order to facilitate the understanding of this resource, we propose a family of information-theoretically constructed preorders on tensors, which can be used to compare tensors with each other and to assess the existence of transformations between them. The construction places copies of a given tensor at the edges of a hypergraph and allows transformations at the vertices. A preorder is then induced by the transformations possible in a given growing sequence of hypergraphs. The new family of preorders generalises the asymptotic restriction preorder which Strassen defined in order to study the computational complexity of matrix multiplication.

\fussy

 We derive general properties of the preorders and their associated asymptotic notions of tensor rank and view recent results on tensor rank non-additivity, tensor networks and algebraic complexity in this unifying frame. We hope that this work will provide a useful vantage point for exploring tensors in applied mathematics, physics and computer science, but also from a purely mathematical point of view.
\end{abstract}

%\tableofcontents

\section*{Introduction}
This work is set on the backdrop of two research topics, each of which has developed the theory of tensors in an information-theoretic limit
$$t \mapsto t \otimes t \otimes \cdots \otimes t,$$
but with a distinct flavour what concerns the notion of tensor product. The first research topic is Strassen's tensor analysis, which he developed in order to treat the complexity of matrix multiplication \cite{burgisser2013algebraic}. Here, the tensor product that governs the asymptotics is the Kronecker product and a crucial quantity is the asymptotic tensor rank. The topic closely connects to recent research on combinatorial problems with a recursive structure, such as the cap set problem, in which the asymptotic subrank is a key tensor parameter \cite{ellenberg}. The second research topic is the study of tensor networks for the description of quantum many-body-physics or the classical simulation of quantum algorithms. Here, the tensor product stands for a partial contraction governed by a lattice, graph or hypergraph \cite{christandl2023resource, cirac2021matrix}. The extreme case, where no contraction is carried out, has a nontrivial tensor parameter associated which measures the asymptotic non-multiplicativity \cite{christandl2018tensor, christandl2019border}.

Whereas in the first topic, following Strassen's association of the matrix multiplication exponent $\omega$ to the 2-by-2 matrix multiplication tensor, the association of asymptotic tensor parameters to the tensor themselves is an important aspect, the situation is markedly different in the second topic. Here, the large objects that have been constructed with tensor product and partial contractions have a physical or computational life of their own as quantum many-body state or computational circuits.

In this work, we change this viewpoint and regard the large objects that can be constructed from a given tensor merely as a lense through which to view the original tensor. As a result we obtain a family of asymptotically defined preorders on tensors as well as associated limiting notions of rank and subrank. Strassen's tensor analysis then becomes a natural special case of our new viewpoint. 

The paper is structured as follows. In the Tensor section, we introduce tensors, the preorders of restriction and degeneration, and remind of polynomial interpolation in this context, which relates the two. In Hypegraph Restriction section, we associate tensors to hypergraphs and introduce the three paradigmatic examples (disjoint, lattice, Strassen), which we focus on. In Asymptotic Hypergraph Restriction section, to each growing sequence of hypergraphs we associate a preorder. We show that asymptotic restriction is a special case and generalise known constructions and obstructions from this case to the more general setting. We conclude with an outlook and open questions in the Conclusion section.

\section*{Tensors} \label{sec:restriction}
Let $t\in \CC^{d_1} \otimes \CC^{d_2} \otimes \cdots \otimes \CC^{d_k}$ be a \emph{tensor} of order $k$ (or $k$-tensor) with \emph{local dimensions} $\{d_j\}_{j=1}^k$. For concreteness we choose the complex numbers as the underlying field, but most statements below remain true for more general fields. Note that we regard the vector spaces as mere vector spaces, i.e.\ without an inner product, as may have been expected in the context of quantum theory, where $k$-tensors are quantum states of $k$ particles (we therefore sometimes use the term \emph{state} instead of tensor, when historically more appropriate). Choosing basis $\{e_{i}^{(j)}\}_{i=1}^{d_j}$ for the $j$'th tensor factor, $t$ can be expressed as 
$$t=\sum_{i_j=1}^{d_j} t_{i_1i_2 \cdots i_k} e^{(1)}_{i_1} \otimes e^{(2)}_{i_2} \otimes \cdots \otimes e^{(k)}_{i_k}.$$
For better readability, we will mostly drop the superscript and write $e_{i_j}$ instead of $e^{(j)}_{i_j}$. Note that $2$-tensors are matrices and that $k$-tensor theory therefore generalizes matrix theory. 

The most basic notion for comparing tensors is that of restriction \cite{burgisser2013algebraic}. 
Given $k$-tensors $t\in \CC^{d_1} \otimes \CC^{d_2} \otimes \cdots \otimes \CC^{d_k}$ and $t'\in \CC^{d'_1} \otimes \CC^{d'_2} \otimes \cdots \otimes \CC^{d'_k}$, we say that $t$ \emph{restricts} to $t'$, and write $t\geq t'$, if there are linear maps $m_j: \CC^{d_j} \rightarrow \CC^{d'_j}$ s.th.
$$t'=(m_1\otimes m_2 \otimes \cdots \otimes m_k) t.$$
Note that we do not require the maps to be invertible (not even when $d_j=d_j'$). 
We say that $t$ and $t'$ are \emph{equivalent}, and write $t \sim t'$, if $t\geq t'$ and $t'\geq t$ and emphasise that this does not necessitate that $t$ and $t'$ are defined with respect to spaces of the same dimension. Indeed, embedding a tensor into a higher dimensional space by padding with zeros will result in an equivalent tensor. It is therefore natural to regard the equivalence class of $t$ under $\sim$, rather than $t$ itself, as a 'tensor', but we will not put much weight on this distinction. 

In summary, restriction is a preorder on the set of $k$-tensors which allows to compare any two $k$-tensors. When viewing tensors as a resource, restriction gives us the tool to transform a tensor $t$ into a tensor $t'$. In quantum information, this way of transforming is equivalent to the notion of transformation under stochastic local operation assisted by classical communication (SLOCC) \cite{bennett-slocc,dur2000three}. When in fixed dimension ($d_j=d'_j$) and when the focus is on the classification of entanglement under SLOCC into SLOCC entanglement classes, one may without loss of generality require the linear maps to be invertible. 

\begin{example} The restriction 
$$e_0 \otimes e_0 \otimes e_1+e_0 \otimes e_1 \otimes e_0+ e_1 \otimes e_0 \otimes e_0\geq e_0 \otimes e_0 \otimes e_0+e_1 \otimes e_1 \otimes e_0$$
of tensors in $\CC^2\otimes \CC^2 \otimes \CC^2$ is obtained by choosing $m_3=e_0 e_0^*$, $m_2=\id$ and $m_1=e_1 e_0^*+ e_0 e_1^*$, where $\{e_i^*\}$ denotes the dual basis. The tensor on the left hand side (LHS) is known as the $\W$-state and the one on the right hand side (RHS) is the $2$-by-$2$ unit matrix viewed as a $3$-tensor, i.e. $(e_0 \otimes e_0 +e_1 \otimes e_1) \otimes e_0$. In quantum information $e_0 \otimes e_0 +e_1 \otimes e_1$ is also known as an Einstein-Podolsky-Rosen (EPR) pair with two levels denoted by $\EPR_2$. We will use the following graphical illustration
\begin{center}
\begin{overpic}[height=1cm,grid=false]{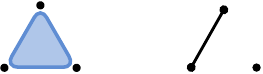}
\put(45,11){\large{$\geq$}}
\put(10,7){$W$} \put(68,13){$2$}\put(120,3){$,$}
\end{overpic}
\end{center}
where each dot represents one of the three particles with associated vector space (here $\CC^2$). The blue triangle represents the $\W$-state and touches all three particles indicating what in quantum information language is known as \emph{genuine multiparticle entanglement}. The edge on the RHS represents the $\EPR$-state ($\sum_i e_i\otimes e_i$) with the number $2$ indicating the number of levels, here $\EPR_2$. The fact that a third particle is present but not touching the rest of the illustration indicates the factorised state $\EPR_2 \otimes e_0$. 

Now let $\eps$ be a non-zero complex number (or a formal variable). We find
   \begin{align*}
       e_0 \otimes e_0 \otimes e_0+& e_1 \otimes  e_1 \otimes e_1 \sim (e_0+\eps e_1)\otimes (e_0+\eps e_1)\otimes (e_0+\eps e_1)-e_0 \otimes e_0 \otimes e_0 \\
       &= \eps (e_0 \otimes e_0 \otimes e_1+e_0 \otimes e_1 \otimes e_0+ e_1 \otimes e_0 \otimes e_0)+O(\eps^2)
       \end{align*} 
       where in the first line we used the invertible matrices $m_1=m_2=m_3=(e_0+\eps e_1)e_0^*+(-e_0)e_1^*$ and in the second line we expanded in powers of $\eps$.
   The initial tensor is known as the unit tensor of size two, often denoted by $\langle 2\rangle$, or, in quantum information, as the Greenberger-Horne-Zeilinger ($\GHZ$) state with two levels. We recognise the final tensor as the $\W$-state. Dividing $m_1$ by $\eps$, we see that 
   $\GHZ_2 \geq \W+O(\eps),$
   i.e.\ the conversion is possible to arbitrary precision, even though it is not exactly possible: $\GHZ_2 \not\geq W$. 
\begin{center}
\begin{overpic}[height=1cm,grid=false]{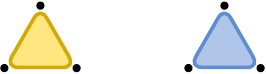}
\put(45,11){\large{$\not\geq$}}
\put(12,7.5){$2$} \put(79,7.5){$W$}
\end{overpic}
\end{center}
This latter statement is relatively easy to prove, yet the first non-trivial application of the lower bound method known as substitution method \cite{burgisser2013algebraic}. 
\end{example}
The previous example shows that the restriction preorder is not closed, therefore calling for the introduction of its closure, in order to make tools from algebraic geometry available.
We therefore say that $t$ \emph{degenerates} to $t'$, and write $t\geqdeg t'$, if, for all $\eps \neq 0$ there are tensors $t_\eps'$ s.th. $t\geq t_\eps'$ and
   $t_\eps' \mapsto^{\eps\mapsto 0} t'$. It turns out that we may equivalently demand that there are maps $\{m_j\}_{j=1}^k$ with entries polynomial in $\eps$, s.th. 
   $t'_\eps=\eps^d t' + \eps^{d+1} t_1+\cdots + \eps^{d+e} t_{e}$, for some $t_j$, furnishing $\geqdeg$ to $\geqdeg^e, \geqdeg_d$ and $\geqdeg_d^e$ as needed \cite[Chapter 15.4]{burgisser2013algebraic}. Note that it is the latter definition that generalises to fields other than $\CC$.

Since degeneration is closed it gives rise to algebraic varieties and non-degeneration $t\not\geqdeg t' $ can therefore be certified with help of polynomial covariants that vanish on $t$, but not on $t'$. More precisely, consider $t, t'$ in the same tensor space (else, enlarge suitably) and note that $t\geqdeg t'$ is equivalent to $\overline{\GL.t}  \supseteq \overline{\GL.t'}$, where $\GL.t$ denotes the orbit $\{(m_1\otimes m_2\otimes \cdots \otimes m_k) t: \forall j, m_j \in \GL_{d_j}(\CC)\}$ and overline denotes the topological closure, which over $\CC$ coincides with the Zariski closure (see e.g. \cite{christandl2019border}). Since $\overline{\GL.t} $ is a $\GL$-invariant algebraic variety, it can be presented as the common zero-set of a finite set of $\GL$-covariants (see e.g. \cite[Section 3 in Supplementary Information]{walter2013entanglement}).  $t' \notin \overline{\GL.t} $ is therefore equivalent to the existence of a  $\GL$-covariant that vanishes on $t$, but not on $t'$. As $\overline{\GL.t} $ is $\GL$-invariant,  $t' \notin \overline{\GL.t} $ is furthermore equivalent to  $\overline{\GL.t'} \not\subseteq \overline{\GL.t} $

\begin{example}
    As an example consider $\CC^2 \otimes \CC^2 \otimes \CC^2$ and note that the statement $\W\geqdeg\GHZ_2$ is equivalent to $\overline{\GL.\GHZ_2}  \supseteq \overline{\GL.\W}$, where $\GL.t$ then equals $\{m_1\otimes m_2\otimes m_3 t: m_i \in \GL_2(\CC)\}$.  Cayley's second hyperdeterminant 
    \begin{align*}
        \mathrm{Det}(t):=t_{000}^2t_{111}^2 &+ t_{001}^2t_{110}^2 + t_{010}^2t_{101}^2 + t_{100}^2t_{011}^2  \\
        & - 2t_{000}t_{001}t_{110}t_{111} - 2t_{000}t_{010}t_{101}t_{111} 
         - 2t_{000}t_{011}t_{100}t_{111}\\ 
         &  - 2t_{001}t_{010}t_{101}t_{110} -2t_{001}t_{011}t_{110}t_{100}
          -2t_{010}t_{011}t_{101}t_{100} \\ 
          & + 4t_{000}t_{011}t_{101}t_{110} + 4t_{001}t_{010}t_{100}t_{111}.
        \end{align*}
    is a polynomial that changes multiplicatively by $(\det m_1)^2(\det m_2)^2(\det m_3)^2$ under restriction and thus stays either zero or non-zero on orbits. In particular, if it is zero, it will remain so on its closure. 
Plugging in the coordinates ($t_{001}=t_{010}=t_{100}=1$ and otherwise zero) for $W$ in the above formula, we see that  $\mathrm{Det}(\W)=0$ and $\mathrm{Det}$ thus vanishes identically on $\overline{\GL.W}$. If $\overline{\GL.W}$ would contain $\overline{\GL.\GHZ_2}$, it would also have to vanish on $\GHZ_2$, but this is not the case, as $\mathrm{Det}(\GHZ_2)=1$ (here $t_{000}=t_{111}=1$ and otherwise zero). Cayley's second hyperdeterminant therefore witnesses that $\W$ does not degenerate to $\GHZ_2$, $\W \not\geqdeg \GHZ_2$. Thus, while
\begin{center}
\begin{overpic}[height=1cm,grid=false]{GHZgeqW.pdf}
\put(45,11){\large{$\geqdeg$}}
\put(12,7.5){$2$} \put(79,7.5){$W$}
\end{overpic}
\end{center}
we have
\begin{center}
\begin{overpic}[height=1cm,grid=false]{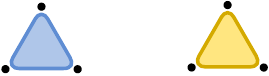}
\put(45,11){\large{$\not\geqdeg$}}
\put(10,7.5){$W$} \put(81,7.5){$2$}
\put(120,3){$.$}
\end{overpic}
\end{center}
\end{example}

The $r$-level GHZ-state on $k$ factors
$$\GHZ_r^{(k)}:=\sum_{i=1}^r e_i^{(1)} \otimes e_i^{(2)} \otimes \cdots \otimes e_i^{(k)}$$ plays a special role in the theory of tensors and is also known as the unit tensor of size $r$ and denoted by $\langle r \rangle$. It may be viewed as a generalisation of the unit matrix to the tensor world, and is canonically obtained from the simple tensor (or product state) $e_1^{(1)} \otimes e_1^{(2)} \otimes \cdots \otimes e_1^{(k)}$ by applying the direct sum operation 
\begin{align*}
    t\oplus t' \in & \CC^{d_1} \otimes \CC^{d_2} \otimes \cdots \otimes \CC^{d_k} \oplus \CC^{d'_1} \otimes \CC^{d'_2} \otimes \cdots \otimes \CC^{d'_k} \\
    & \subseteq (\CC^{d_1} \oplus \CC^{d'_1})  \otimes (\CC^{d_1} \oplus \CC^{d'_1})  \otimes \cdots \otimes (\CC^{d_1} \oplus \CC^{d'_1})  
\end{align*}
where the inclusion is by padding with zeros. It is then clear that 
$$\GHZ_r^{(k)}\sim \bigoplus_{i=1}^r e_1^{(1)} \otimes e_1^{(2)} \otimes \cdots \otimes e_1^{(k)},$$
where $e_1^{(j)}$ spans $j$'th factor $\CC$ and thus 
$e_1^{(1)} \otimes e_1^{(2)} \otimes \cdots \otimes e_1^{(k)}$ spans
$\CC \otimes \CC \otimes \cdots \otimes \CC$. The $i$-direct sum component then corresponds to
$e_i^{(1)} \otimes e_i^{(2)} \otimes \cdots \otimes e_i^{(k)}$ of the GHZ-state defined above.

When viewing tensors as a resource, the GHZ-state is the natural `currency'. The `cost' of a tensor is then the size of a unit tensor required to obtain $t$ (now keeping $k$ implicit), 
$$R(t):=\min \{r:\GHZ_r \geq t\}$$
a quantity identical to the well-known \emph{tensor rank} of $t$.
Likewise, the `value' of $t$, the largest GHZ-state we may obtain from $t$,
$$Q(t):=\max \{r:t \geq \GHZ_r \},$$ 
is known as the \emph{subrank} of $t$. Whereas clearly $R(\GHZ_2)=Q(\GHZ_2)=2$, the example above implies $R(W)=3$ and $Q(W)=1$, exhibiting irreversibility in tensor transformations.

Viewing physical, computational and mathematical objects as resources and studying their free transformation with associated costs and values is common and often implicit to a subject. Considering an explicit \emph{resource theory} is well-known from thermodynamics and a focal point of entanglement theory \cite{horodecki-review}. In the present context a tensor resource theory was considered as SLOCC entanglement theory \cite{dur2000three} and then connected to algebraic complexity 
\cite{chitambar2008}. We summarise the above in the following resource theory for $k$-tensors. 
\begin{resource}[$\geq$]
The resource theory of tensors under restriction is given by:
    \begin{itemize}
    \item (resource) $t$ a $k$-tensor
    \item (transformation) restriction $\geq$
    \item (unit) unit tensor or GHZ-state $\GHZ_r$
    \item (cost) tensor rank $R(t):=\min \{r:\GHZ_r \geq t\}$
    \item (value) subrank $Q(t):=\max \{r:t \geq \GHZ_r \}$
\end{itemize}
\end{resource}
Alternatively, one can consider a resource theory of degeneration, where, in place of the preorder restriction, $\geq$ one uses the preorder of degeneration $\geqdeg$. The corresponding cost and value in this resource theory are known as \emph{border rank} $\underline{R}$ and \emph{border subrank} $\underline{Q}$. Since $\geq$ implies $\geqdeg$, degeneration is weaker than restriction with the implied relations
$$R(t) \geq \underline{R}(t) \geq \underline{Q}(t) \geq Q(t).$$ 
Note that this chain of inequalities collapses to the usual rank for the matrix case $k=2$. Generally, $\geqdeg$ is not much weaker than $\geq$ as can be seen in the following useful lemma, which is proved with help of polynomial interpolation.
\begin{lemma}[\cite{bini1980approximate}]\label{lemma:interpolation}
    For $t, t'$ $k$-tensors with $t \geqdeg^e_d t'$, we have
    $$\bigoplus_{i=1}^{e+1} t \geq t'.$$ The statement is also true with $e+1$ replaced by $\binom{k-1+d}{k-1}$.
\end{lemma}

The first consequence of Lemma \ref{lemma:interpolation} is $$R(t) \leq (e+1)\underline{R}(t)$$ for $t \geqdeg^e t'$, with the full potential of the lemma becoming clear after the next section.

The study of every resource theory requires constructions giving explicit transformations and obstructions thereof. A complete understanding beyond the simplest cases is often unattainable. This is no different in the resource theory of tensors for $\geq$ or $\geqdeg$, where only the matrix case ($k=2$), the matrix pencil case ($3$-tensors where $d_1=2$) and a few small cases (e.g.\ $d_1=d_2=d_3=3$, $d_1=d_2=d_3=d_4=2$) are well-understood \cite{Lan:TensorBook}. The problem is in general NP-hard as the computation of tensor rank is \cite{tensor-rank-np}.

It is therefore remarkable that a significant treatment of larger structured tensors has been obtained in different contexts ranging from algebraic complexity to quantum many-body physics. In the following section we will explain how to build large structured tensors by placing smaller tensors on edges of a hypergraph. 

\section*{Hypergraph Restriction} \label{sec:hypergraph}
Whereas general tensors of large order and dimensions are unwieldy objects, there are ways of constructing powerful structured tensors from smaller ones. Key elements in the constructions are the tensor product operation, grouping (or flattening) of vector spaces and the partial contraction with tensors of smaller order. In the context of quantum many-body physics and quantum computation this leads to matrix product states (MPS) and projected entangled pair states (PEPS), higher order tensors of small local dimensions \cite{orus2019tensor, cirac2021matrix}. As the names suggest the focus is here on combining pairs, i.e.\ matrices or $2$-tensors, but recently, more general structures have been considered by combining smaller tensors to larger entanglement structures \cite{christandl2020tensor, molnar2018generalization, xie2014tensor, christandl2023resource}. MPS are also known as tensor trains (TT) \cite{oseledets} and have applications as numerical mathematical tool in a range of scientific disciplines from engineering and data analysis to the life sciences \cite{kolda}. 

The contractions with smaller tensors considered in this context is a special case of the notion of restriction and therefore naturally included in our framework, and consequently omitted as a building principle for structured tensors. In the context of algebraic complexity theory \cite{burgisser2013algebraic}, the tensor product is usually used to obtain tensors of the same low order $k$, but in higher local dimensions. The purpose of this section is to explain a single framework that exhibits both tensor networks and algebraic complexity as important special cases (\cite{christandl2023resource} building on \cite{christandl2019tensor, vrana2017entanglement, christandl2020tensor}). We start with a basic example displaying different ways of building larger structured tensors.

\begin{example} \label{ex:tensorproduct}
    Let $t\in \CC^{d_1}\otimes \CC^{d_2}\otimes \CC^{d_3} $ and $t'\in \CC^{d_1'}\otimes \CC^{d_2'}\otimes \CC^{d_3'} $ be $3$-tensors. Then 
    $t \otimes t' \in \CC^{d_1}\otimes \CC^{d_2}\otimes \CC^{d_3} \otimes\CC^{d_1'}\otimes \CC^{d_2'}\otimes \CC^{d_3'} $ is naturally a $6$-tensor, but by grouping tensor factors, may also be regarded as a $3$-tensor
    $t \otimes t' \in ( \CC^{d_1}\otimes\CC^{d_1'}) \otimes ( \CC^{d_2} \otimes \CC^{d_2'}) \otimes ( \CC^{d_3} \otimes \CC^{d_3'}) $
     or anything in between, i.e.\ as a 4- or 5-tensor, as shown in the illustration (from left to right as 6-, 5-, 4- and 3-tensor):
\begin{center}
\begin{overpic}[height=1cm,grid=false]{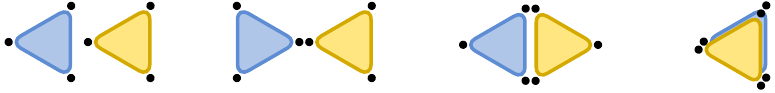}
\put(110,2){.}
\end{overpic}
\end{center}
The properties of the four tensors will in general be different, a first glimpse of which is given by the tensor rank. The smallest example is provided by the W-state, where
$$R(W)^2=9>R(W\otimes W)=8>R(W\boxtimes W)=7 $$
with $R(W\otimes W)$ to be interpreted as the 4-, 5- or 6-tensor \cite{chen2018tensor, yu2010tensor, christandl2018tensor, christandl2023resource}, graphically
\begin{center}
\begin{overpic}[height=1.3cm,grid=false]{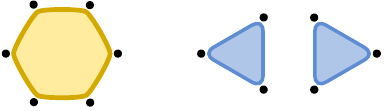}
\put(38,12){$\geq$}
\put(14,12){$8$} 
\put(58,12){$W$}
\put(83,12){$W$}
\put(110,2){.}
\end{overpic}
\end{center}
The sign $\boxtimes$ was used to indicate that the final tensor is a 3-tensor. More generally, it is used to indicate the \emph{Kronecker tensor product} which groups to two $k$-tensors into a $k$-tensor.

When considering the threefold tensor product of $\EPR_d$, again interpreted in the different groupings, we obtain, for instance,
\begin{center}
\begin{overpic}[height=1cm,grid=false]{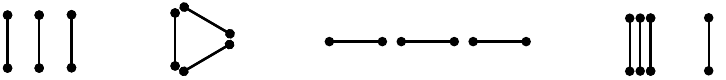}
\put(92,4){\large{$\sim$}}
\put(99,4){$d^3$}
\put(110,2){,}
\end{overpic}
\end{center}
where the tensor associated to the triangle graph (indicated by $\triangle$) is the famous $d$-by-$d$ matrix multiplication tensor $${\EPR_d}_{\triangle}:=\sum_{i_1, i_2, i_3=1}^d (e_{i_1} \otimes e_{i_2}) \otimes (e_{i_2} \otimes e_{i_3}) \otimes (e_{i_3} \otimes e_{i_1}), $$
which is also denoted by $\langle d, d, d\rangle$ or $\mathrm{MaMu}(d)$ \cite{burgisser2013algebraic,christandl2020tensor}.
Using more pairs, we can build arbitrary graph tensors \cite{christandl2019tensor} including the entanglement structures of MPS 
\begin{center}
\begin{overpic}[height=0.3cm,grid=false]{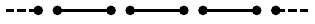}
\end{overpic}
\end{center}
or PEPS \cite{christandl2020tensor}
\begin{center}
\begin{overpic}[height=3cm,grid=false]{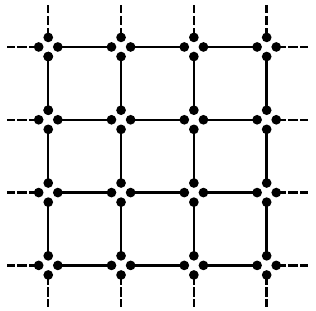}
\put(110,2){.}
\end{overpic}
\end{center}

\end{example}

A useful relation between the tensor product and the direct sum for $k$-tensors is
$$\bigoplus_{i=1}^r t\sim \GHZ_r^{(k)} \boxtimes \ t$$
which implies
$$\GHZ_{rr'}^{(k)}=\GHZ_{r}^{(k)}\boxtimes \GHZ_{r'}^{(k)}.$$
Note that the same relations hold true for EPR-states as  $\GHZ_{r}^{(2)}=\EPR_{r}$.

The common trait in the examples is that we have associated structured tensors to graphs or hypergraphs motivating the following general definition. 
\begin{definition}{\cite[Definition 11]{christandl2020tensor}} \label{def:entanglement-structure}
Let $H=(V, E)$ be a (directed) hypergraph with vertex set $V$ and hyperedge set $E$.\footnote{A hyperedge $e=(v_1, v_2, \ldots, v_k)$ is an ordered set of distinct vertices. We allow hyperedges to be repeated. Sometimes we will indicate the number of hyperedges $n$ of $H$ in subscript, i.e.\ $H_n=(V_n, E_n)$.}
To each hyperedge $e$, let $t_e\in \bigotimes_{v \in e} \CC^{d_v^{(e)}}$ be an associated tensor.
We then define the tensor 
$$t_H:=\bigotimes_{e\in E} t_e \in \bigotimes_{v \in V} \left(\bigotimes_{e: v \in e}\CC^{d_v^{(e)}}\right)$$
where the bracket indicates that the tensor spaces associated the same vertex are combined into one tensor factor. 
We call $t_H$ an \emph{entanglement structure} and note that it is a tensor of order $|V|$.
\end{definition}
 Note that an entanglement structure $t_H$ is derived both from the hypergraph $H=(V, E)$ and the tensors $t_e$ which are associated to the hyperedges $e\in E$. In many cases the $t_e$'s can be reconstructed from $t_H$ and $H$, e.g.\ in the case where each hyperedge appears only once and no permutation of its vertices are present. In the context of GHZ-states on the edges, entanglement structures have also been called (hyper)graph states \cite{christandl2019tensor}. 
 The following illustration shows a tensor associated to a hypergraph with four edges of sizes two, three, three and four on seven vertices. 
\begin{center}
\begin{overpic}[height=2cm,grid=false]{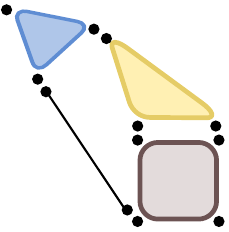}
\end{overpic}
\end{center}  
The different colours are meant to indicate the different tensors. For simplicity of the illustration we decided not to indicate that the hyperedges are directed. 

When the hypergraph is uniform (meaning all hyperedges are of the same size), it might happen that all $t_e$ are the same and equal to $t$. When clear from the context, we might thus sometimes associate $t_H$ directly to a tensor $t$ without mentioning the intermediate map from edges to tensors. 

Since we will later focus on families of hypergraphs we now  introduce three paradigmatic families that exhibit the richness of the subject.
\begin{itemize}
    \item (Disjoint) 
    Given hypergraphs $H$ and $H'$, we denote their disjoint union (graph sum) by $H \oplus H'$. For the $\nu$-fold sum, we write $\nu \cdot H$. Consider now $H$ as a hypergraph on $k$ vertices with a single hyperedge of size $k$. Then we define $H_{\disj_n}:= \nu \cdot H$, i.e.\ the hypergraph on $kn $ vertices with $n$ disjoint edges of size $k$, here illustrated with $k=3$ and $n=8$,
\begin{center}
\begin{overpic}[height=1.5cm,grid=false]{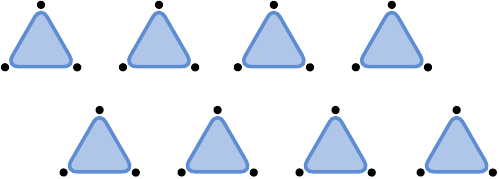}
\put(110,2){.}
\end{overpic}
\end{center}    
\item (Lattice) Consider $H_{\lattice_n}$, a patch of $n$ hyperedges of a (hyper)lattice. We illustrate with the vertices arranged in the  triangular lattice with hyperedges on every second plaquette of the lattice,
\begin{center}
\begin{overpic}[height=2cm,grid=false]{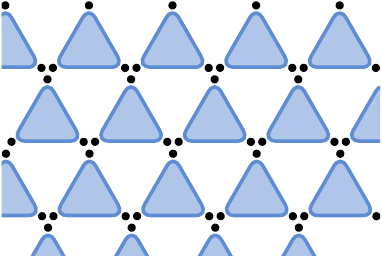}
\put(110,2){.}
\end{overpic}
\end{center}  
    \item (Strassen) Consider $H_{\Str_n}$, the hypergraph on $k$ vertices with $n$ occurrences of the same hyperedge of size $k$, which is relevant to Strassen's asymptotic restriction,
\begin{center}
\begin{overpic}[height=1cm,grid=false]{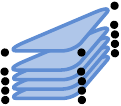}
\put(110,2){.}
\end{overpic}
\end{center}  
\end{itemize}

Clearly, if $t_e \geq t'_e$ for all $n$ edges of $H$, then $t_H \geq t_H'$ as can be read off from the definition of $t_H$ (Definition \ref{def:entanglement-structure}), since grouping tensor factors enlarges the maps used for restrictions from 
$\otimes_{e\ni v} m_v^{(e)}$ to arbitrary linear maps taking $\otimes_{e\ni v} \CC^{d_v^{(e)}} $ to $\otimes_{e\ni v} \CC^{d_v^{'(e)}}$. 
Equivalently, we may first place the tensors $t_e$ on the disconnected hypergraph $H_{\disj_n}$ with edges $e$. Note that 
$t_{H_{\disj_n}}\geq t'_{H_{\disj_n}}$ implies $t_{H}\geq t'_{H}$, since $H$ can be obtained from $H_{\disj_n}$ by grouping vertices, i.e.\ by partitioning the vertices or by applying a hypergraph homomorphism\footnote{For us a \emph{hypergraph homomorphism} is a map from the vertex set of one hypergraph to another such that each hyperedge maps to a hyperedge.}. More generally, if $H$ can be obtained from $\tilde{H}$ by grouping of vertices, then $t_{\tilde{H}} \geq t'_{\tilde{H}}$ implies $t_{H} \geq t'_{H}$, since again the grouping of vertices enlarges the maps that can be used to effect a restriction. Restrictions $t_{H} \geq t'_{H}$ have recently been studied in their own right in the context of lattices, motivated by the study of tensor networks \cite{christandl2020tensor,christandl2023resource}.

We now focus on the case of uniform hypergraphs, i.e. hypergraphs where each edge has the same size, say size $k$, and the situation in which we associate to each edge the same $k$-tensor $t$. The resulting entanglement structure $t_H$ may now be viewed as an $H$-dependent tool with which to study $t$ itself. This is a key change of perspective which we wish to emphasize in this work. In particular, this allows us to consider the study of the restriction $t_H \geq t'_H$ as a lens through which we view the comparison of $t$ and $t'$. We therefore introduce the following preorders on tensors. 
\begin{definition}
Let $H$ be a $k$-uniform hypergraph and $t, t'$ two $k$-tensors. We say that $t$ \emph{$H$-restricts} to $t'$, and write $t \geq_H t'$, whenever $t_H \geq t'_H$. 
\end{definition}
$H$-restriction is weaker than restriction, i.e.\ $t\geq t'$ implies $t \geq_H t'$, or $$\geq \implies \geq_H$$ for short. More generally, for $H$ that can be obtained from $\tilde{H}$ by grouping, we have 
$$\geq_{\tilde{H}} \implies \geq_H.$$
For patches of lattices (or more generally hypergraphs) which can be folded onto the Strassen hypergraph (i.e.\ the vertices can be grouped such that we obtain the Strassen hypergraph) we have 
\begin{align}\label{eq:implies} 
\geq_{\disj_n} \implies \geq_{\lattice_n} \implies  \geq_{\Str_n}
\end{align}
and we will for the simplicity of the following discussion focus on such lattices.

$H$-restriction is strictly weaker than restriction, meaning that there are $t, t'$ s.th.~$t \not\geq t'$ but $t_H\geq t_H'$, precisely when $H$ is not (Berge) acyclic \cite{landsberg2012geometry, christandl2023resource}. Even if $H$ is acyclic and $H$-restriction therefore the same as restriction, the study of the associated tensor parameters is still meaningful as they remain non-trivial as we had seen in Example \ref{ex:tensorproduct}. We therefore introduce the \emph{$H$-rank} of $t$ as 
$R_H(t):=R(t_H)$ and the \emph{$H$-subrank} of $t$ as 
$Q_H(t):=Q(t_H)$ which, just as restriction, weakens under grouping
$R_{\tilde{H}}(t)\geq R_{H}(t)$ and 
$Q_{H}(t) \geq Q_{\tilde{H}}(t)$.
 $H$-rank and $H$-subrank are also natural restriction-monotone functions, which can serve as obstructions for $H$-restriction. Similarly, we may introduce $H$-degeneration,  $H$-border rank and $H$-border subrank and relate them to their restriction versions through Lemma \ref{lemma:interpolation}. 

In the following we will consider the behaviour of $H_n$-restriction for large-$n$, leading to new notions of asymptotic preorders on tensors.

\section*{Asymptotic Hypergraph Restriction}
\label{se:asy}
Strassen's asymptotic preorder $\geqas$, which is defined as 
$$t \geqas t' \quad \mathrm{ if } \quad \bigoplus_{i=1}^{2^{o(n)}}t^{\boxtimes n} \geq t'^{\boxtimes n},$$ 
plays an instrumental role in a systematic understanding of the matrix multiplication exponent $\omega$ \cite{burgisser2013algebraic, christandl2023universal}. The famous conjecture 
$\omega=2$, in particular, has the compact formulation $\GHZ_4 \geqas {\EPR_2}_{\triangle}.$

By taking an appropriate large-$n$ limit of $\geq_{H_n}$, we will introduce a set of preorders generalizing asymptotic restriction. The small direct sum in Strassen's preorder has the purpose of making the definition robust against minor changes, like swapping restriction against degeneration.\footnote{Sometimes the equivalent definition $t \geqas t' $ if $t^{\boxtimes n+o(n)} \geq t'^{\boxtimes n}$ is used. We choose the small direct sum instead, as it is better suited for the generalisation presented in this work. Equivalence follows as, on the one hand, the rank of $t$ is finite, and, on the other hand, asymptotic subrank of $t$ is non-trivial as soon as $t$ is not a tensor product of smaller order tensors.\cite{strassen1988asymptotic}} This aspect will also be important for our more general considerations and we therefore propose the following preorders on tensors. 

\begin{definition} \label{def:preorder}
Let $\bH=\{H_n\}_{n\in \NN}$ be a sequence of $k$-uniform hypergraphs and $t, t'$ be $k$-tensors. We say that \emph{$t$ ${\bf H}$-restricts to $t'$} and write
$$t \geqas_{\bf H} t' \quad \mathrm{if} \quad \bigoplus^{2^{o(n)}}_{i=1} t  \geq_{H_n} t'.$$    
\end{definition}

Asymptotic restriction is recovered when choosing $H_n=\Str_n$, i.e.\ $\geqas$ equals $\geqas_{\bStr}$, where $\bStr=\{\Str_n\}_{n\in \NN}$. The indifference in this case to the use of degeneration instead of restriction is proved with polynomial interpolation. Whereas we leave the general question open whether $\bH$-restriction and $\bH$-degeneration are identical, we show it for cases, where $\bH$ has the following property. 
\begin{definition}
    Let $\bH=\{H_n\}_{n\in \NN}$ be a sequence of $k$-uniform hypergraphs. We say that $\bH$ is \emph{subadditive} if for $n_0(n) \in o(n)\cap \omega(1)$, there are $r(n) \in o(n)$ and $\nu(n)$ s.th.\ for all $n\in \NN$, $H_{n}$ can be obtained by grouping the vertices of the disjoint union of $\nu$ copies of $H_{n_0}$ and some $\tilde{H}_r$.
\end{definition}

Note that $\nu \leq n/n_0$.
 The examples in this manuscript are subadditive as the following illustrates.
\begin{example}
    Consider a $d$-dimensional lattice, where $H_n$ is a hypercubic patch of the lattice (obtained by cutting the infinite lattice with a hypercube and by adding a few edges in order to make all natural numbers $n$ possible). Fix a small hypercubic patch with side length $N_0$ and $n_0=N_0^d$ edges and fill the larger $H_n$ with $\nu$ (which is roughly $n/n_0$) copies of $H_{n_0}$. Choose $N_0\in o(n^{\frac{1}{d}})\cap \omega(1)$. The remaining $r$ edges satisfy $r\in o(n)$, since they arise as a surface term. 
\end{example}

Before showing that $\bH$-degeneration and $\bH$-restriction are the same, we discuss the following consequence of Lemma \ref{lemma:interpolation}, which can be found in \cite{christandl2020tensor} stated without the notion of a hypergraph preorder.
\begin{lemma}\label{lemma:interpolation-structure}
    Let $H_n$ be a $k$-uniform hypergraph and let 
    $t \geqdeg t'$ be two $k$-tensors. Then
    $$\bigoplus_{i=1}^{O(n)} t\geq_{H_n} t'$$
\end{lemma}

The merely linear direct sum enables unexpected restrictions on spaces which have exponentially growing dimension: Whereas $\GHZ_2 \not\geq_{H_n} W$ on any three-uniform hypergraph $H_n$ that folds onto the Strassen hypergraph\footnote{The argument here is that $R(W^{\boxtimes n})>2^n$ \cite{chen2010,zuiddam2017}. Note that \cite{christandl2023resource} gives a different argument that holds for another set of hypergraphs. Both sets contain the triangular and kagome lattices, which we will consider later as examples.}, since $\GHZ_2 \geqdeg W$ we find $\bigoplus_{i=1}^{O(n)} \GHZ_2 \geq_{H_n} W$ with help of Lemma \ref{lemma:interpolation-structure}. In other words, the restriction is enabled by a global GHZ-state $\GHZ_{O(n)}^{|V_n|}=(\GHZ_2^{|V_n|})^{\boxtimes (\log n+O(1))}$ of only a logarithmic number of qubits locally and is therefore a powerful tool in constructions. As mentioned earlier, Lemma \ref{lemma:interpolation} also enables us to show that $\bH$-restriction equals to $\bH$-degeneration in cases of interest to us.

\begin{theorem} \label{th:subadditivie}
   Let $\bH=\{H_n\}_{n\in \NN}$ be a subadditive sequence of $k$-uniform hypergraphs. Then $\bH$-restriction and $\bH$-degeneration are the same, i.e.
   $$\bigoplus^{2^{o(n)}}_{i=1} t  \geq_{H_n} t' \quad \mathrm{iff} \quad
   \bigoplus^{2^{o(n)}}_{i=1} t  \geqdeg_{H_n} t'.$$    
\end{theorem}
\begin{proof}
    Since $\geq$ implies $\geqdeg$, the first statement implies the second. We now expand the second statement into
    \begin{align}
        \label{eq:bound1}    \GHZ_2^{\boxtimes f(n)} \boxtimes \ t_{H_n} \geqdeg^{e(n)} t'_{H_n},
    \end{align}
    where $\GHZ_2$ extends over all vertices of $H_n$, $f(n)\in o(n)$ and $e(n)$ is the $n$-dependent error degree of the degeneration.
    We will also use the bound 
    \begin{align}
        \label{eq:bound2}
        \GHZ_2^{\boxtimes O(n)} \geq t'_{H_n}
    \end{align}
which comes from the fact that $t'$ has finite tensor rank and that $H_n$ has $n$ hyperedges.
Fix now $n_0, n, r$ s.th. $n=\nu n_0+r$. Taking $\nu $ copies of \eqref{eq:bound1} with $n_0$ instead of $n$ we find
    $$\GHZ_2^{\boxtimes\nu f(n_0)} \boxtimes \ t_{\nu \cdot H_{n_0}} \geqdeg^{\nu e(n_0)} t'_{\nu \cdot H_{n_0}}.$$
By Lemma \ref{lemma:interpolation} this implies 
    $$\GHZ_2^{\boxtimes (\nu f(n_0)+\lceil\log_2 (\nu e(n_0)+1)\rceil)} \boxtimes \ t_{\nu \cdot H_{n_0}} \geq t'_{\nu \cdot H_{n_0}}.$$
We now tensor \eqref{eq:bound2} with $r$ instead of $n$ to this inequality and obtain 
    $$\GHZ_2^{\boxtimes (\nu f(n_0)+\log_2 (\nu e(n_0)+1))+O(r))} \boxtimes \ t_{\nu \cdot H_{n_0}\oplus H_r} \geq t'_{\nu \cdot H_{n_0}\oplus H_r}.$$
Set now $n_0\equiv n_0(n)=\min\{\lfloor \sqrt{n}\rfloor, \max \{m: e(m) \leq 2^{\sqrt{n}}\}\}$ and note that it satisfies the assumptions $n_0(n) \in o(n)\cap \omega(1)$ as required by the subadditivity definition. Since $\bH$ is subadditive, we now find $r(n)\in o(n)$ and know that $\nu \cdot H_{n_0}\oplus H_r$ can be grouped to $H_n$. This implies
    $$\GHZ_2^{\boxtimes (n f(n_0)/n_0+\log_2 (n/n_0 e(n_0)+1)+ o(n))} \boxtimes \ t_{H_{n}} \geq t'_{H_{n}}.$$
 Since $n_0(n) $ is a growing function of $n$, the first part of the exponent is $o(n)$. Since $e(n_0)\leq 2^{\sqrt{n}}$ the remainder of the exponent is also $o(n)$. We therefore see that the exponent is $o(n)$ and the first condition in the statement is fulfilled.
\end{proof}

We leave it as an open question to settle whether this theorem extends to all $\bH$ and if not to investigate the novel asymptotic degeneration preorders in their own right.

Since grouping weakens $H_n$-restriction, we obtain 
$$\geqas_{\bdisj} \implies \geqas_{\blattice} \implies \geqas_{\bStr}. $$
Note that constructions for asymptotic restrictions on the left imply constructions on the right and obstructions for the right imply obstructions for the left. 

Obstructions are typically obtained from monotone functions. Cost and value are in general resource theories the canonical monotone functions. For the restriction resource theory these were directly defined via the preorder, e.g.
$$R(t) := \min \{r: \GHZ_r \geq t\}.$$
In an asymptotic context one might be tempted to introduce 
$$\min \{r: \GHZ_r \geqas t\}$$
as the cost, but this turns out not to be such a useful quantity, e.g.\ because it only assumes integral values. 
The cost is therefore better defined as the regularization or amortization of the tensor rank ($\liminf$ can be replaced by $\lim$ here)
\begin{align*}\asymprank(t) :=\liminf_{n\rightarrow \infty}R(t^{\boxtimes n})^{\frac{1}{n}}
&= \liminf_{n\rightarrow \infty}R_{\Str_n}(t)^{\frac{1}{n}}.
\end{align*}
$\asymprank(t)$ is monotone under asymptotic restriction and captures the matrix multiplication exponent via $\omega= \log_2 \asymprank({\EPR_2}_\triangle)$. 

The natural generalization to $\bH$-rank reads
\begin{align*}\asymprank_{\bH}(t)& :=\liminf_{n\rightarrow \infty} R_{H_n}(t)^{\frac{1}{n}}
\end{align*}
and it is easy to see that  the $\bH$-rank is an $\bH$-restriction monotone. Likewise, we can define the $\bH$-subrank 
$$\asympsub_{\bH}(t):=\limsup_{n\rightarrow \infty}  Q_{H_n}(t)^{\frac{1}{n}}$$
as the value. $\bH$-subrank is also an $\bH$-restriction monotone.

Similarly, we may define $\bH$-border rank and $\bH$-border subrank. Whereas we leave the relation of the latter to the $\bH$-subrank in the unclear, we show that the former two coincide for subadditive $\bH$ by a similar argument to the above.
\begin{theorem} Let $\bH$ be subadditive. Then
    $\bH$-border rank equals $\bH$-rank.
\end{theorem}
\begin{proof}
By definition, $\bH$-border rank is smaller than $\bH$-rank. 
 Let $$f(n)=\lceil\log \underline{R}(t_{H_n})\rceil.$$ Then by definition
    $$\GHZ_2^{\boxtimes f(n)} \geqdeg^{e(n)} t'_{H_n},$$
    where $\GHZ_2$ extends over all vertices of $H_n$.
    We will also use the bound 
    $$\GHZ_2^{\boxtimes O(n)} \geq t'_{H_n}.$$
        Fix now $n_0, \nu, r$ s.th. $n=\nu n_0+r$. Taking $\nu $ copies of the first bound with $n_0$ instead of $n$ we find
    $$\GHZ_2^{\boxtimes \nu f(n_0)} \geqdeg^{\nu e(n_0)} t'_{\nu \cdot H_{n_0}}.$$
By Lemma \ref{lemma:interpolation} this implies 
    $$\GHZ_2^{\boxtimes (\nu f(n_0)+\lceil \log_2 (\nu e(n_0)+1)\rceil)} \geq t'_{\nu \cdot H_{n_0}}.$$
We now tensor to this inequality the second bound with $r$ instead of $n$ and obtain 
    $$\GHZ_2^{\boxtimes (\nu f(n_0)+\log_2 (\nu e(n_0)+1))+O(r))}  \geq t'_{\nu \cdot H_{n_0}\oplus H_r}.$$
Let now $n_0\equiv n_0(n)=\min\{\sqrt{n}, \max \{m: e(m) \leq 2^{\sqrt{n}}\}\}$ and note that it is $o(n)\cap \omega(1)$ as required by the subadditivity definition. 
Since $\bH$ is subadditive we now find $r(n)$ and know that $\nu \cdot H_{n_0}\oplus H_r$ can be grouped to $H_n$. This implies 
    $$\GHZ_2^{\boxtimes (n f(n_0)/n_0+\log_2 (n/n_0 \cdot e(n_0)+1))+o(n))} \geq t'_{H_{n}}.$$
Since $n_0(n) $ is a growing function of $n$ and $f(n)\in O(n)$ we see that the first part of the exponent has the same limit as $f(n)/n$, as desired. Since $e(n_0)\leq 2^{\sqrt{n}}$ the remaining terms in the exponent are $o(n)$. Overall, we find that $\bH$-rank is smaller than $\bH$-border rank.
   \end{proof}

By definition we have $\asymprank_{\bH}(t) \geq  \asympsub_{\bH}(t)$. As before, grouping vertices in the hypergraph leads to relations among the associated quantities, in this case the asymptotic ranks:
$$\asymprank_{\bdisj}(t)\geq \asymprank_{\blattice}(t)\geq \asymprank_{\bStr}(t). $$

We summarise the arising novel asymptotic resource theories:
\begin{resource}[$\geqas_{\bH}$]
The resource theory of tensors under $\bH$-restriction is given by:
    \begin{itemize}
    \item (resource) $t$ a $k$-tensor
    \item (transformation) $\bH$-restriction $\geqas_{\bH}$
    \item (unit) unit tensor or GHZ-state $\GHZ^{(k)}_r$
    \item (cost) $\bH$-rank $\asymprank_{\bH}(t)$
    \item (value) $\bH$-subrank $\asympsub_{\bH}(t)$
\end{itemize}
\end{resource}

Note that both $\bH$-rank and $\bH$-subrank are merely subnormalised, i.e. 
$$R_\bH(\GHZ^{(k)}_r)\leq r \ \mathrm{and} \ Q_\bH(\GHZ_r)\leq r.$$
Strict inequality for the rank occurs when the hypergraph is acyclic for an extensive number of edges. This is so, for instance, in the matrix multiplication case, which can be formulated as follows: let $k=2$ and consider $H_n$ be the graph on three vertices with roughly $\frac{n}{3}$ edges between each pair of vertices. Then $\asymprank_{\bH}(\EPR_2)=2^{\frac{\omega}{3}}<2$ where we note that $\EPR_2=\GHZ^{(2)}_2$. Subrank is mostly strictly subnormalised and even equals to one when the hypergraph is mostly disconnected, for instance in the case $H_\bdisj$ considered earlier.

In the following we discuss the three examples of $\bdisj$-, $\blattice$- and $\bStr$-restriction in more detail and place previous work in the context of the new asymptotic preorders put forward in this work. We always begin the discussion with  constructions and subsequently elaborate on tools for obstructions.

\paragraph{Disjoint}

 In \cite{christandl2018tensor} it was shown that $R(W \otimes W)<R(W)^2$ and hence that the rank under the ordinary tensor product is not multiplicative in general (see Example \ref{ex:tensorproduct}). The same phenomenon was observed for border rank \cite{christandl2019border} and in order to study the amortized quantification of this phenomenon the \emph{tensor asymptotic rank}
 $$R^{\otimes}(t):=\lim_{n\rightarrow \infty} R(t^{\otimes n})^{\frac{1}{n}}$$ was introduced. 
 In our notation this quantity equals
 $$\asymprank_{\bdisj}(t):= \lim_{n\rightarrow \infty} R_{\disj_n}(t)^{\frac{1}{n}}$$
and is associated with the preorder  $\geqas_{\bdisj}$.
 
A nontrivial construction is obtained from Lemma \ref{lemma:interpolation-structure} when applied to the disjoint hypergraph.
\begin{theorem}[Construction \cite{christandl2018tensor}] \label{th:border}
$$\underline{R}(t) \geq \asymprank_{\bdisj}(t)$$
\end{theorem}

Since $\underline{R}(W)=2$, we have in particular that $\asymprank_{\bdisj}(\W)=2$, thereby determining the asymptotic manifestation of the non-multiplicativity of the $W$-state. 
In \cite{christandl2019border}) it was shown that the inequality in Theorem \ref{th:border} can be strict which thereby opened the search for upper bounds on $\asymprank_{\bdisj}(t)$ beyond the border rank. 

The main lower bound method for border rank, generalized flattenings (see \cite{garg2019} and references therein), work in this setting as they are multiplicative under the (disjoint) tensor product. In order to see this, let $F: \CC^{d_1} \otimes \CC^{d_2} \otimes \cdots  \otimes \CC^{d_k}\rightarrow  \CC^{D_1} \otimes \CC^{D_2} $ be a linear map from tensors to matrices. Then define
$$\Young_{F}(t):=\frac{\rk(F(t))}{\max_s \rk F(s)},$$
where the maximisation is over simple tensors (i.e.\ $s=\alpha_1 \otimes \alpha_2 \otimes \cdots \otimes \alpha_k$) and $\rk$ is the matrix rank. We then have the following theorem, which extends the mentioned lower bound $\underline{R}(t) \geq R_F(t)$.
\begin{theorem}[Obstruction \cite{christandl2018tensor,christandl2019border}]
   Let $t$ and $F$ be as above. Then, $$\asymprank_{\bdisj}(t) \geq R_F(t).$$
\end{theorem}

\begin{example} \label{Ex:flattening}
In \cite{christandl2019border} a flattening lower bound of 4.5 was obtained for a specific tensor of border rank at most 5. This in particular showed that the border rank indeed equals 5, but more so that 
    $$\asymprank_{\bdisj}(t) \geq 4.5.$$ Border rank is non-multiplicative for this tensor with the best upper bound on $\asymprank_{\bdisj}(t)$ being 4.746368884 obtained by using 7 copies \cite{christandl2019border}.
\end{example}

Note that $Q_\bdisj$ always equals to one and is thus not an interesting quantity to consider.

\paragraph{Lattice}

The study of lattice conversions was first considered in \cite{christandl2020tensor,christandl2023resource} and was motivated by the use of tensor networks and for the description of many-body quantum systems: Let $\psi_n$ be a sequence of $n$-body quantum states that can be represented by an underlying \emph{entanglement structure} given by $t'_{\lattice_n}$, i.e.\ $t'_{\lattice_n} \geq \psi_n$. If now $t_{\lattice_n} \geq t'_{\lattice_n}$ then also $t_{\lattice_n} \geq \psi_n$ and we conclude that $\psi_n$ can also be represented by the entanglement structure $ t_{\lattice_n}$. Depending on the task at hand, converting to a different entanglement structure can have theoretical and numerical benefit for the understanding of many-body physics. Concrete lattice conversions based on polynomial interpolation were introduced in 
\cite{christandl2020tensor} and the importance of a small additional direct sum was noted. In \cite{christandl2023resource} lattice conversions were developed into a full resource theory for tensor networks. 

Here we want to change this viewpoint. Instead of using $t$ and $t'$ to construct entanglement structures $t_{\lattice_n}$ and $t'_{\lattice_n}$, and to study their conversion under restriction as $O(n)$-tensors in their own right, we want to shift the focus back to $t$ and $t'$ and only use $t_{\lattice_n}$ and $t'_{\lattice_n}$ as vehicles to inform our understanding of $t$ and $t'$. That is, we want to consider $\geq_{\lattice_n}$ ($\lattice_n$-restriction), $\geqas_{\blattice}$ ($\blattice$-restriction) as well as the corresponding ranks as objects and tensor parameters associated to $t$ (and $t'$).

The results that have been obtained in the context tensor networks can then be formulated in the following tensor-centric way.

\begin{theorem}[Construction \cite{christandl2020tensor}] \label{th:lattice}
    $t \geqdeg t'$ implies $$t \geqas_{\blattice} t'.$$
\end{theorem}

\begin{example} A physically motivated example is the PEPS presentation of the resonating valence bond state (RVB) on the kagome lattice \cite{RVB}. In the kagome lattice regular triangles surround regular hexagons. Whereas previously, a bond dimension of three was obtained, in \cite{christandl2020tensor} it was shown that ${\EPR_2}_{\triangle}\geqdeg\lambda$, where $\lambda =e_1 \wedge e_2 \wedge e_3 + e_3 \otimes e_3 \otimes e_3$. Placing the tensors on the triangles of the lattice, by Theorem \ref{th:lattice}, one finds
${\EPR_2}_{\triangle}\geqas_{\bkagome} \lambda$
or, more precisely,
$$\bigoplus_{i=1}^{O(n)}{\EPR_2}_{\triangle} \geq_{\kagome_n} \lambda.$$
This result indeed requires the small direct sum, as it was shown in \cite{christandl2023resource}
that for all $n$
$${\EPR_2}_{\triangle}\not\geq_{\kagome_n} \lambda$$
for a kagome lattice with boundary. In conclusion, physical properties of the RVB state can be computed when having access to a linear number of parallel bond dimension two computations \cite{christandl2020tensor}.
\end{example}

Even though a small global $\GHZ$-state was used, the theorem is still based on a plaquette-by-plaquette degeneration. In \cite{christandl2023resource} it was shown that asymptotic lattice restrictions are possible beyond this construction, i.e.\ there are cases where $t\not\geqdeg t'$, but $t \geqas_{\blattice} t'$, exhibiting that the $\blattice$-restriction preorder is distinct from the degeneration preorder.

We will now turn to the discussion of obstructions, which are again obtained by utilizing generalized flattenings. Since the lattice is more connected than the disjoint hypergraph, not all flattenings will be multiplicative and we need to restrict our use somewhat, but luckily not by much.

We call a generalized flattening a \emph{Young flattening}\footnote{We recommend \cite[Section 8.2.2.]{landsberg2017geometry} and note that we here, due to our application, only consider the application of the map on the third tensor factor.} of $\CC^{d_1}\otimes \CC^{d_2}\otimes \CC^{d_3}$ if it has the following structure
$$F(t)=\left( \id \boxtimes Y\right) (t),$$
where the identity acts on the first two tensor factors and $Y$ maps the last tensor factor equivariantly into a set of matrices\footnote{This is what connects to representation theory, therefore the name Young attached to it.}: 
$$Y: \CC^{d_3}\rightarrow \CC^{d_3'} \otimes \CC^{d_3''}.$$
 That is, there are $\GL_{d_3}$-representations $a$ and $b$ of dimensions $d'_3$ and $d''_3$, respectively, s.th.
$$Yg=\left( a(g)\otimes b(g) \right) Y$$
for all $g \in \GL_{d_3}$.
The combined map $F$ then maps the tensor space into a matrix space:
$$F=\id \boxtimes Y: \CC^{d_1}\otimes \CC^{d_2}\otimes \CC^{d_3}\rightarrow (\CC^{d_1}\otimes \CC^{d_3'} ) \otimes ( \CC^{d_2}\otimes \CC^{d_3''}).$$

Following \cite{christandl2023resource} we note that, since the Young flattening only acts non-trivially on one of the tensor spaces, we can group the others and still preserve the multiplicativity of the flattening bound that we had discussed earlier for the disjoint hypergraph. More precisely, consider a lattice which can be folded onto a fan. With this we mean that there is a grouping of the vertices s.th.\ the lattice turns into a fan after grouping. In Figure 7 (a) of \cite{christandl2023resource}, a triangular lattice is folded onto a fan resulting in a six-fold covering of the fan. 

The following theorem is inspired by \cite[Section V.B.]{christandl2023resource}

\begin{theorem}[Obstruction] \label{th:Young}
    Consider a lattice that can be folded onto a fan with $c$-fold covering and $F$ a Young flattening of $(\CC^{d_1})^{\otimes c}\otimes (\CC^{d_2})^{\otimes c}\otimes (\CC^{d_3})^{\otimes c}$. Then $t \geqas_{\blattice} t'$ implies 
    $ \rk F(t^{\boxtimes c}) \geq \rk F(t'^{\boxtimes c})$.
\end{theorem}
\begin{proof}
$t \geqas_{\blattice} t'$ implies $T  \geqas_{\bfan} T'$ for $T:=t^{\boxtimes c}$ (and likewise with $'$) which is defined as 
$$\bigoplus^{2^{o(n)}}_{i=1} T_{\fan_n} \geq T'_{\fan_n}$$
or equivalently 
\begin{align}
    \label{eq:fan}
\sum^{2^{o(n)}}_{i=1} T_{i} = T'_{\fan_n} \ \mathrm{ with }  \ 
T_{\fan_n} \geq T_{i} \ \forall  \ i.
\end{align}
Consider the map $F^{(n)}$ given by 
\begin{align*}
  \left( \id  \boxtimes Y^{\otimes n}\right): &\left((\CC^{d_1})^{\otimes n}\right)\otimes \left((\CC^{d_2})^{\otimes n} \right) \otimes (\CC^{d_3})^{\otimes n} \\
  &\rightarrow \left((\CC^{d_1}\otimes \CC^{d_3'} )^{\otimes n} \right) \otimes \left( ( \CC^{d_2}\otimes \CC^{d_3''})^{\otimes n} \right).
  \end{align*}
Applying the linear map to both sides of \eqref{eq:fan} implies
$$\sum_{i=1}^{2^{o(n)}} F^{(n)}( T_{i}) = F^{(n)}(T'_{\fan_n}).$$ 
Writing $T_i=A_i \otimes B_i\otimes \left(\otimes_j c^{(j)}_i \right) T_{\fan_n}$ and using the covariance property of the Young flattening we find
\begin{align*}
    F^{(n)}( T_{i})&= F^{(n)}(  A_i \otimes B_i\otimes \left(\otimes_j c^{(j)}_i \right) T_{\fan_n})\\
    &= A_i \otimes B_i \otimes \left(\otimes_j a(c^{(j)}_i)  \otimes b(c^{(j)}_i) \right) F^{(n)}(T_{\fan_n})\\
    &= \left( A_i \otimes (\otimes_j a(c^{(j)}_i) )\right) \otimes \left(B_i \otimes (\otimes_j b(c^{(j)}_i) )\right) F^{(n)}(T_{\fan_n})\\
    & \leq F^{(n)}( T_{\fan_n}),
 \end{align*}
where in the last step we note that we have a restriction of matrices. 
 We find
 $$\bigoplus_{i=1}^{2^{o(n)}}F^{(n)}( T_{\fan_n})\geq F^{(n)}( T'_{\fan_n}).$$
Since 
$$F^{(n)}(T_{\fan_n})=F(T)^{\boxtimes n} $$
and similarly for $T'$, we have 
$$\bigoplus_{i=1}^{2^{o(n)}}F(T)^{\boxtimes n}\geq F(T')^{\boxtimes n}.$$
We now apply the matrix rank to this restriction. Since it is additive under direct sum, multiplicative under the Kronecker tensor product and monotone under restriction we find
\begin{align*}
2^{o(n)}(\rk F(T))^n  \geq (\rk F(T'))^n.
\end{align*}
Taking the $n$'th root and the large-$n$ limit concludes the proof.\end{proof}

\begin{example}
We will consider Young flattenings $F^{(c)}$ that are $c$-fold tensor products of Young flattenings $F$, in which case the condition 
$$ \rk F(t^{\boxtimes c}) \geq \rk F(t'^{\boxtimes c})$$ is equivalent to 
 $\rk F(t) \geq \rk F(t')$.
$F$, in turn, we take of the special form of a Koszul flattening. Koszul flattenings have been successfully used to obtain lower bounds for border rank, or obstructions to 
\cite{landsberg2015new}
\begin{align}
    \label{eq:koszul}
   \GHZ_{r} \geqdeg t.
\end{align} 

\sloppy 
Here, for a given $p$, $\CC^{d'_3}=\bigwedge^{p+1}(\CC^{d_3})$ and $\CC^{d''_3}=\bigwedge^{p}(\CC^{d_3})^*$ as $\GL_{d_3}$-representations. 
\fussy
Since by the Pieri rule
$\bigwedge^p(\CC^{d_3})\otimes (\CC^{d_3}) = \bigwedge^p(\CC^{d_3}) \oplus \ldots$ is multiplicity-free, there is a (up to scale) unique intertwiner $Y$ given by 
\begin{align*}
    Y:\CC^{d_3} & \rightarrow \bigwedge^{p+1} (\CC^{d_3}) \otimes \bigwedge^{p}(\CC^{d_3})^*\\
\ket{v} & \mapsto \sum_w (\ket{w} \wedge \ket{v}) \otimes \ket{w^*},
\end{align*}
where the sum extends over a basis with elements $\ket{w}$ of $\bigwedge^p(\CC^{d_3})$ with $\ket{w^*}$ the dual basis and we note that $\ket{w} \wedge \ket{v} \in \bigwedge^{p+1}(\CC^{d_3})$.

Consider now the case $d_1=d_2=d_3=3$ and the tensor 
$t=\GHZ_3$ and $t'$ as the tensor in \cite[Prop. 3.1]{christandl2019border}, which we had already considered in Example \ref{Ex:flattening}. As the proof in this reference shows, $\rk F(t')=9$, whereas $\rk F(\alpha\otimes \beta\otimes \gamma)=2$ which implies $\rk F(\GHZ_3)\leq 6$ and thus shows via the theorem that 
$\GHZ_3 \not\geqas_{\blattice} t'.$ 
\end{example}

In the special case $t=\GHZ_r$ that we considered, this statement can be strengthened \cite[Theorem 11]{christandl2023resource} to give $\GHZ_4 \not\geqas_{\blattice} t'$. If $t'$ is the matrix multiplication tensor $\EPR_{D \triangle}$, obstructions for ranges of $r$ and $D$ can be obtained. 

\begin{remark}
What makes generalized flattenings seem unnatural in our context is the fact that they depend on the embedding dimension, which means that the tensor needs to be regarded as an element in the vector space rather than in relation to the equivalence class $\sim$. In order to extend the Young flattening $F$ of $\CC^{d_1}\otimes \CC^{d_2}\otimes \CC^{d_3}$ to arbitrary dimensions $\CC^{d_1}\otimes \CC^{d_2}\otimes \CC^{D_3}$ with $D_3 \geq d_3$ one may consider defining \footnote{An alternative would be only to maximise over linear maps with image in $\CC^{d_3}$.}
$$F(t):= \sup_{t \geqdeg t' \in \CC^{d_1}\otimes \CC^{d_2}\otimes \CC^{d_3}} \rk F(t').$$
It would be interesting to explore this viewpoint on Young flattenings in future work. 
\end{remark}

\paragraph{Strassen}

Several of the ideas presented above in the context of hypergraphs, disjoint graphs and lattices apply to the Strassen hypergraph and were, in fact conceived in this context. This applies especially to Lemma \ref{lemma:interpolation-structure} in its application to border rank
$$\underline{R}(t) \geq \asymprank(t)$$
where
$$\asymprank(t):=\lim_{n\rightarrow \infty}R(t^{\boxtimes n})^{\frac{1}{n}},$$
which was used in order to derive upper bounds on the matrix multiplication exponent $\omega$ beyond Strassen's original rank-based algorithm \cite{burgisser2013algebraic}. 

A tool that goes beyond the polynomial interpolation that is behind this lemma and which integrally uses the Kronecker tensor product $\boxtimes$ is Strassen's laser method \cite{burgisser2013algebraic}. Here, it is central that asymptotic degeneration admits multiplicative cancellation,
\begin{align}
    \label{eq:cancel}
t \boxtimes  c \geqas t' \boxtimes  c \implies t \geqas t'
\end{align}
as we now will see. Consider an intermediate tensor $\iota$ for which transformations
$$\GHZ_\alpha \geqas \iota $$
and 
$$\iota \geqas {\EPR_2}_{\triangle} \boxtimes \ \GHZ_\beta,$$ 
can be found. Then note that $\GHZ_\alpha\sim \GHZ_{\alpha/\beta} \boxtimes  \GHZ_\beta$
if $\alpha$ is divisible by $\beta$ and set
$t=\GHZ_{\alpha/\beta} $, $c=\GHZ_\beta$ and $t'={\EPR_2}_{\triangle}$. Combining this we find
$$\GHZ_{\alpha/\beta}\boxtimes  \GHZ_\beta \sim \GHZ_\alpha \geqas \iota \geqas {\EPR_2}_{\triangle}\boxtimes  \GHZ_\beta.$$
which by \eqref{eq:cancel} implies 
$$\GHZ_{\alpha/\beta}\geqas {\EPR_2}_{\triangle}.$$
The following theorem, which is the essence of Strassen's laser method, records this computation.

\begin{theorem}[Construction] Let $\iota$ be a tensor. 
$\GHZ_\alpha\geqas \iota$ together with $\iota \geqas {\EPR_2}_{\triangle} \boxtimes \ \GHZ_\beta$ implies
$\omega \leq \log_2(\alpha/\beta)$.
\end{theorem}

The construction of an \emph{intermediate tensor} $\iota$ for which good values for $\alpha$ and $\beta$ can be obtained is based on the idea of placing EPR pairs in a larger outer structure which may be thought of as a scaffolding. Coppersmith and Winograd used for $\iota$
$$cw_q=\sum_{i=1}^q e_i \otimes e_i \otimes e_0 +e_0 \otimes e_i \otimes e_i +e_i \otimes e_0 \otimes e_i,$$
which has $q$-level EPR pairs embedded in a $W$ structure. Similar to the good border rank upper bound for the W-state one derives a good border rank bound also for $cw_q$ and thus a good, i.e. small, value for $\alpha$.

Since $\EPR$-pairs are the blocks of $cw_q$, the blocks in tensor powers are tensor products of $\EPR$-pairs when placed in different directions. In the language of algebraic complexity theory, these are rectangular matrix multiplication tensors. For the $m$'th power ($m$ divisible by three) these are $$\langle q^m, 1, 1\rangle, \langle q^{m-1}, q, 1 \rangle, \ldots, \langle q^{m/3}, q^{m/3}, q^{m/3}\rangle.$$ Each block appears many times, restricting to a fixed type, say the equally weighted one, results in the desired
$ \bigoplus \langle q^m,q^m,q^m \rangle$ with $\beta$ being determined by $q, m$ and the size of the direct sum, which is related to a variant of the subrank of the outer structure. The subrank of the $W$-state was determined by Strassen as $\asympsub (W) =2^{h(\frac{1}{3})}$ for $h$ the binary entropy and corresponds to the size of the direct sum.\footnote{The words 'essence of the method' and 'variants of the subrank' refer to the fact that those restrictions in the laser method cannot be arbitrary, but need to preserve the block structure, i.e.~be monomial restrictions.}

Employing sophisticated analysis of the block structure one obtains the currently best bounds of around $2.37$ on $\omega$ \cite{duan2022faster}. It remains open whether the lower bound of 2 can be achieved. Likely new new intermediate tensors would be required as there are barriers for achieving exponent 2 with $cw_q$, $q>2$  \cite{barriers-alman,barriers-christandl}.

Whereas Strassen's work focused on tensors of order three due to the motivation from the matrix multiplication problem, the techniques are often more general and have been investigated in particular for tensors $\GHZ_H$, i.e. tensors obtained by placing GHZ-states on the hyperedges \cite{vrana2017entanglement, christandl2019asymptotic, christandl2019tensor}.

We now turn to the discussion of obstructions for asymptotic restriction. This discussion must take its starting point in Strassen's remarkable characterisation theorem, which posits the existence of a complete set of monotones. In the following we use the symbol $F$ in a different way from earlier.

\begin{theorem}[Characterisation \cite{strassen1988asymptotic}] \label{th-Strassen} Let $t,t' $ be $k$-tensors. Then
    $$t \geqas_{\bStr} t' \quad \mathrm{iff} \quad  F(t) \geq F(t')$$ holds for all Strassen-$F$s, i.e. for all real-valued functions $F$, defined on all $k$-tensors, satisfying
    \begin{itemize}
        \item monotonicity under restriction: $s\geq s'$ implies $F(s)\geq F(s') $
        \item normalization on $\GHZ_r$:$F(\GHZ_r)=r$
        \item multiplicativity under Kronecker tensor product: $F(s\boxtimes s')=F(s)F(s')$
        \item additivity under direct sum: $F(s \oplus s')=F(s)+F(s')$
    \end{itemize}
\end{theorem}

Each Strassen-$F$ is thus an obstruction to asymptotic restriction as $F(t')>F(t)$ implies $t\not\geqas_{\bStr} t'$. It is remarkable that a complete set of such structured obstructions exist and one might wonder what the Strassen-$F$s are and if one can construct them. Easy to construct are the so-called gauge points obtained by grouping of all but tensor factor $j$, to obtain matrix $t_j$ and then considering the matrix rank $\rk(t_j)$. Whereas Strassen was able to construct further (Strassen-)$F$s for subrings of tensors, which include the matrix multiplication tensor, general Strassen-$F$s were unknown until recently, where the \emph{quantum functionals} $F_\theta(t)=2^{E_\theta(t)}$ with
$$E_\theta(t):= \sup_{t\geqdeg t'} \sum_j \theta_j H(t_j),$$ 
were constructed \cite{christandl2023universal}. Here, $H(t_j)$ is the Shannon entropy of the squared singular values of the matrix $t_j$ and $\theta$ is a probability distribution on the set $\{1, 2, \ldots, k\}$.

\begin{theorem}[Obstruction\cite{christandl2023universal}]
    The quantum functionals $F_\theta$ are \linebreak Strassen-$F$s. 
\end{theorem}

It remains open whether these are all Strassen-$F$s; if true, this would imply $\omega=2$.

\section*{Conclusion}
\label{sec:con}
We have introduced a new family of asymptotically defined preorders on tensors as well as their natural associated asymptotic notions of rank and subrank (Definition \ref{def:preorder}). The family refines Strassen's asymptotic restriction, asymptotic rank and asymptotic subrank, which form an extreme special case. 

In the context of this special case and motivated by the recent progress on the cap set problem, the interest in multiparticle entanglement and renewed interest in the matrix multiplication problem, several new asymptotic tensor parameters such as the asymptotic slice rank \cite{tao} and the quantum functionals \cite{christandl2023universal} have been introduced. In addition to specific results, even their global structure has been investigated (see \cite{discreteness} and references therein). It will be interesting to view these tensor parameters, and the reason for their introduction, through the lens of hypergraphs and their associated preorders, ranks and subranks (Open problem 1). 

Remarkable about Strassen's asymptotic restriction is existence of a complete set of obstructions (Theorem \ref{th-Strassen}). 
As the new notions of restrictions are weaker than asymptotic restriction (they imply the latter), it begs the question of whether a similar characterization, with a larger set of functions, exists (Open problem 2). Although even in Strassen's case, it is not clear what the set precisely is, we have identified with the Young flattenings a first family of new functions in the case of certain lattices (Theorem \ref{th:Young}).

On the constructive side, we have employed polynomial interpolation, which had already seen a generalization from the use in matrix multiplication to non-additivity \cite{christandl2018tensor} and tensor networks \cite{christandl2020tensor}, to the understanding of the robustness of the new preorders in cases that include lattices (Theorem \ref{th:subadditivie}). We leave it as a main open question to understand whether this robustness persists also in sequences of denser hypergraphs (Open problem 3). 

We have included the description of the laser method in this manuscript in order to highlight that, in contrast to polynomial interpolation, there are constructions that integrally use properties of the asymptotic preorder and whose generalization to the new asymptotic preorders will require an adequate limitation of the method (Open problem 4). 

Abstracting from the concrete remarks on constructions and obstructions above, we hope in particular that the study of sequences of denser hypergraphs will illuminate the study of tensors through their associated preorders (Open problem 5). 

Finally, one may view our results as attaching preorders and functionals to hypergraphs and sequences thereof which obey monotonicity under hypergraph homomorphisms. This opens the possibilty to learn about hypergraphs through the study of tensors. 

\section*{Acknowledgements}

\sloppy
We thank Freek Witteveen for support with the graphical tensor presentations and the Section of Mathematics at the University of Geneva for their hospitality. We acknowledge financial support from the European Research Council (ERC Grant Agreement No.~818761), VILLUM FONDEN via the QMATH Centre of Excellence (Grant No.~10059) and the Novo Nordisk Foundation (grant \linebreak NNF20OC0059939 ‘Quantum for Life’). We also thank the National Center for Competence in Research SwissMAP of the Swiss National Science Foundation (Grant No.~205607).
\fussy

\theendnotes

\bibliographystyle{unsrt}

\providecommand{\noopsort}[1]{}\providecommand{\singleletter}[1]{#1}%

\end{document}